\documentclass{l4dc2025}
\title[]{An Active Parameter Learning Approach to \\ The Identification of Safe Regions}
\usepackage{times}
\usepackage{wrapfig}
\usepackage{mathrsfs}  
\usepackage{mathtools}
\author{%
 \Name{Aneesh Raghavan} \Email{aneesh@kth.se}\\
 \addr DCS Division, KTH, Royal Institute of Technology, Stockholm
 \AND
 \Name{Karl {Henrik Johansson}} \Email{kallej@kth.se}\\
 \addr DCS Division, KTH, Royal Institute of Technology, Stockholm
}
\usepackage{algorithm}
\usepackage{algorithmic}
\allowdisplaybreaks

\begin{document}
\maketitle
\vspace{-1.4cm}
\begin{abstract}%
We consider the problem of identification of safe regions in the environment of an autonomous system. The  environment is divided into a finite collections of Voronoi cells, with each cell having a representative, the Voronoi center. The extent to which each region is considered to be safe by an oracle is captured through a trust distribution. The trust placed by the oracle conditioned on the region is modeled through a Bernoulli distribution whose the parameter depends on the region. The parameters are unknown to the system. However, if the agent were to visit a given region, it will receive a binary valued random response from the oracle on whether the oracle trusts the region or not. The objective is to design a path for the agent where, by traversing through the centers of the cells, the agent is eventually able to label each cell safe or unsafe. To this end, we formulate an active parameter learning problem with the objective of minimizing visits or stays in potentially unsafe regions.  The active learning problem is formulated as a finite horizon stochastic control problem where the cost function is derived utilizing the  large deviations principle (LDP). The challenges associated with a dynamic programming approach to solve the problem are analyzed. Subsequently, the optimization problem is relaxed to obtain single-step optimization problems for which closed form solution is obtained. Using the solution, we propose an algorithm for the active learning of the parameters. A relationship between the trust distributions and the label of a cell is defined and subsequently a classification algorithm is proposed to identify the safe regions. We prove that the algorithm identifies the safe regions with finite number of visits to unsafe regions. We demonstrate the algorithm through an example.
\end{abstract}
\section{Introduction}
\subsection{Motivation}
Safety critical systems are those systems whose failure leads to loss of life, infrastructure damages and or damages to environment, \cite{knight2002safety}. Study of safety critical robotic systems has received significant attention in the literature, for e.g. see \cite{goble2010control}, \cite{fan2020bayesian}, \cite {li2024geometry}. Safe exploration using reinforcement learning \citep{moldovan2012safe,sui2015safe,dalal2018safe}  and model predictive control \citep{koller2018learning} involve safety constraints either in the dynamics of the system or the state space of the system. However identification of the safe regions in the state space or the environment of the agent has received little attention in the literature. 

For example, consider a surveillance problem where the objective of an UAV is identify safe regions to land in a new environment, \cite{patterson2014timely}. Apriori, the UAV does not posses the required knowledge. However, the UAV could scan specific regions of the new environment and transmit the sensor information to a human, who could respond to the UAV letting know the if the human trusts the particular region or not. Through repeated scanning, the UAV could develop a trust distribution for each region. Based on the trust distribution, the UAV could classify if a region is safe or not. One of the main question that arises is, what path should the UAV follow while scanning the regions and developing trust distributions. The UAV could chose a random path, however it would be ideal to design a path that optimizes its learning of the trust distributions and visits the unknown unsafe regions minimally. 

Adaptive sampling for classification has been studied in \citep{djouzi2022new,singh2017sequential,shekhar2021adaptive},  where sequential sampling algorithms are presented to enhance the learning process.  In \cite{ding2021adaptive}, adaptive sampling has been applied to hyperspectral image classification leading to improvement from state of the art. Learning unknown environment is a crucial part of marine robotics.  Adaptive sampling methods have been used to survey and learn about algal bloom,  water quality models, etc, for e.g. see \cite{zhang2007adaptive}, \cite{bernstein2013learning}, \cite{stankiewicz2021adaptive}, and \cite{fossum2020compact}.   However, the perspective of safety has not been considered, i.e, most of those algorithms are not suitable for the problem being considered. 
\vspace{-0.5cm}
\subsection{Problem Considered}
The problem setup is as follows. We consider the environment of an autonomous system (agent), for e.g. a octagonal subset of 2D Euclidean space as depicted in Figure \ref{Figure 2}. The environment is partitioned into a finite collection of Voronoi cells (regions). Each cell is represented using its center. In Figure \ref{Figure 2}, the environment is partitioned into $9$ cells where the center of each cell is represented by blue dots. A path for the agent is a sequence of states of the agent, where each state is one of the centers of the Voronoi cells. The path of the agent is modeled through a Markov decision process (MDP). At a given state, the observation collected by the agent is $\{0,1\}$ valued random variable where $1$ signifies that the oracle trust the region and $0$ signifies that the oracle does not trust the region. Using the observations, the parameters, i.e, the probability that the observation is $1$ given the cell is $j$ (depicted as $p_j$ in Figure \ref{Figure 2}) is to be estimated for all $j$. 
\begin{wrapfigure}{R}{0.5\textwidth}
\vspace{-0.8cm}
  \begin{center}
    \includegraphics[width=0.48\textwidth]{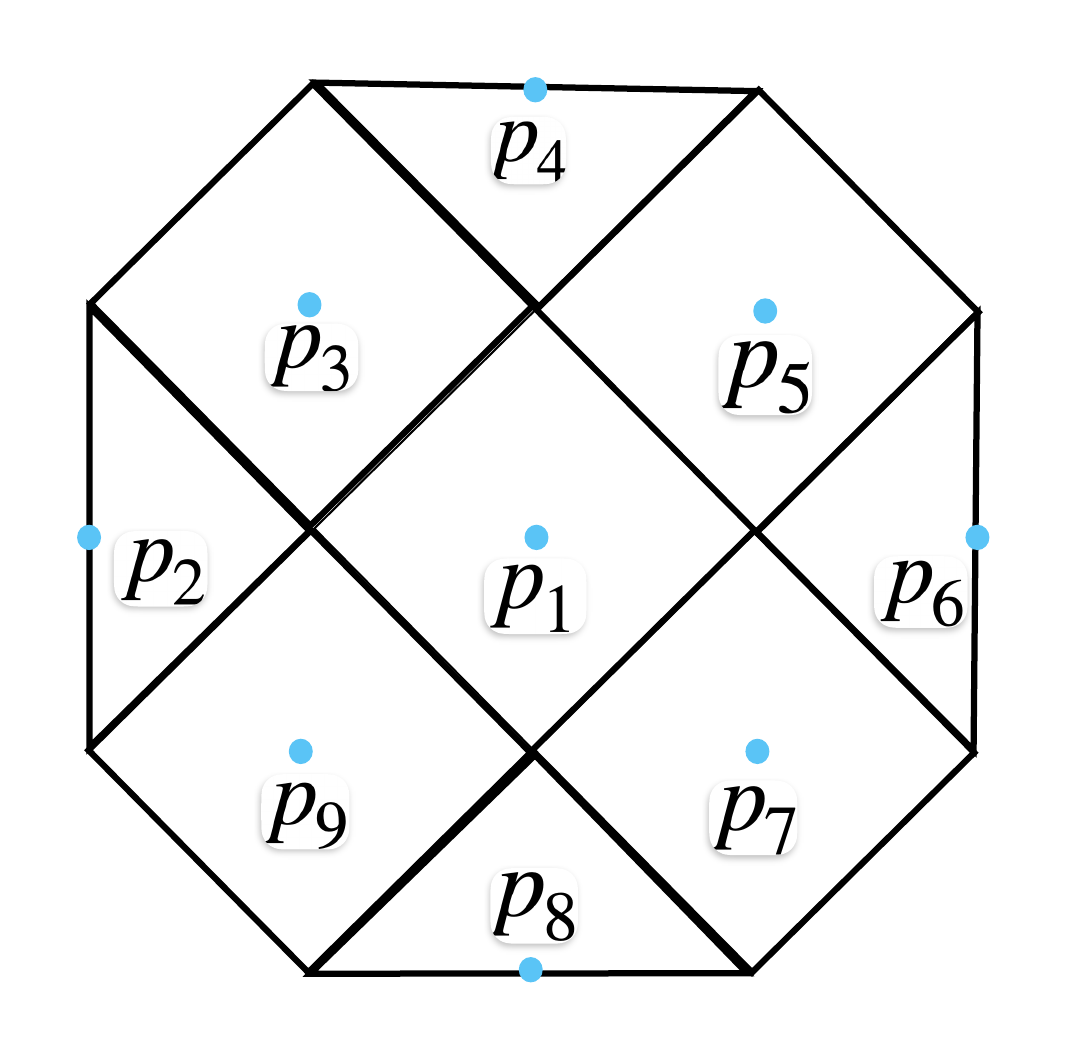}
  \end{center}
  \vspace{-1cm}
  \caption{Voronoi Diagram for the Environment of the Autonomous Agent}
  \label{Figure 2}
\vspace{-1.5cm}
\end{wrapfigure}

The objectives are as follows: (i) to design a path for the agent so that the desired probabilities are estimated quickly while the potentially unsafe regions are visited minimally. (ii)  To define the relationship between safety and estimated probabilities and identify which cells are safe.  
\vspace{-0.5cm}
\subsection{Contributions}
The contributions are as follows. We define asymptotic independence of the observations collected by the agent and utilize the same to characterize the rate of decay of the probability of error in the estimation of the mean probability of the random observation being $1$, where the expectation is with respect to the stationary distribution of the MDP.  The rate of decay is characterized using the large deviations principle, \cite{touchette2009large,varadhan2016large}. The path planning problem is formulated as a finite horizon stochastic control problem whose cost function is the estimated probability of error in the mean of the parameters. The dynamic programming solution is analyzed and due to its computational inefficiency, we relax the original problem to one step optimization problems. Using the concavity of the cost function, we find closed form expressions for the optimal solution of the one step optimization problems. Using the solutions, we propose an active parameter  learning algorithm. Using the median of the estimates of the parameters, we define the safety of the cells and a present classification algorithm to label the cells. 
\vspace{-0.5cm}
\subsection{Outline}
The outline of the paper is as follows. In the next section, we describe the problem setup and define the active parameter learning problem. In Section \ref{Analysis of the Problem}, we analyze the formulated problem and define its relaxation which is computationally tractable. In Section \ref{Algorithm}, we present the learning and the classification algorithm using the solution of the relaxed problem. We present  a numerical example in Section \ref{Example} and conclude with some comments and directions for future work in \ref{Conclusion and Future Work}.
\vspace{-0.5cm}
\section{Problem Formulation}
In this section, we describe the problem setup and formulate the active parameter learning problem as a stochastic control problem. 
\vspace{-0.5cm}
\subsection{Problem Setup}\label{Problem Setup}
We consider the euclidean space, $\mathbb{R}^d, d=2 ,3$. The space is divided into a finite collection of Voronoi cells. Given the centers  of these Vornoi regions, $\mathcal{X} = \{\bar{x}_{j}\}^{m}_{j=1}$, the Vornoi regions are then defined as, 
\begin{align*}
E_{j} = \{x \in \mathbb{R}^d : || x - \bar{x}_{j} || <  || x - \bar{x}_{k} ||, \; \forall k \in \{1, \ldots , m\}, k \neq j \}
\end{align*}
The state of the agent at time instant $n$ is denoted by, $X_n$, $X_{n}(\omega) \in \{\bar{x}_{j}\}^{m}_{j=1}, \; \forall n \in \mathbb{N}$. The observation collected by the agent at state $X_n$ is $Y_n$, where $Y_{n}(\omega) \in \{0,1\}, \; \forall n \in \mathbb{N}$. 

For this experiment, the sample space for the agent at stage $n$ is $\Omega_{n} = \{\mathcal{X} \times \{0,1\}\}^n$. The $\sigma$ algebra of events is $\mathcal{F}_{n} = 2^{\Omega_n}$. Let $\mathcal{D}(\mathcal{X})$ denote the space of probability distributions on $\mathcal{X}$. The decision policy of the agent at $n$ is $\pi_{n}(\cdot): \{\mathcal{X} \times \{0,1\} \} \to \mathcal{D}(\mathcal{X})$. $\pi_{n}(x_{n}, y_{n})[\bar{x}_j]$ is the probability that $X_{n+1}= \bar{x}_j$ given that $X_{n}= x_{n}$ and $Y_{n}= y_{n}$. That is, we assume that the decision of the agent at instant $n$ depends only on the current state and observation rather than the entire past. The conditional distribution of $Y_{n}$ given $\{X_{j} = x_j\}^{n}_{j=1}, \{Y_{j} = y_j\}^{n}_{j=1} $ is defined as, 
\begin{align*}
\pi(Y_{n} = 1 | \{X_{j} = x_j\}^{n}_{j=1}, \{Y_{j} = y_j\}^{n}_{j=1}  ) = p , \pi(Y_{n} = 0 | \{X_{j} = x_j\}^{n}_{j=1}, \{Y_{j} = y_j\}^{n}_{j=1}  ) = 1-p , 
\end{align*}
where $p$ is unknown. Thus, the probability distribution at stage $n$, $\mathbb{P}_{n}$ is 
\begin{align*}
\mathbb{P}_{n}(\{x_j, y_j\}^n_{j=1}) = \prod^n_{j=1}\Big(\mathbf{1}_{\{y_{j}=1\}}p + \mathbf{1}_{\{y_{j}=0\}}(1-p)\Big)\Big(\pi_{j-1}(x_{j-1},y_{j-1})[x_{j}]\Big),
\end{align*}
where $\pi_{0}(x_0,y_0)$ with some abuse of notation is the distribution of the initial state of the agent which is assumed be to known. The above product measure is defined in the sense of \textit{Carathéodory-Hahn Theorem}, i.e., a suitable premeasure is defined on the semiring of rectangle subsets of $\Omega_n$. Then it is extended to a measure on the sigma algebra $\mathcal{F}_n$. Thus, the probability space at stage $n$ is $(\Omega_n, \mathcal{F}_n, \mathbb{P}_n)$. By the \textit{Kolmogorov Extension Theorem}, the sequence of consistent measures, $\{\mathbb{P}_{n}\}$, can be extended to $\mathbb{P}$ on $(\Omega, \mathcal{F})$, where $\Omega = \{\mathcal{X} \times \{0,1\}\}^{\mathbb{N}}$ and $F$ is the $\sigma$ algebra of cylindrical subsets of $\Omega$. The objective of the agent is to estimate $p$. 

The best estimate of $p$, given $\{Y_{j}\}^{n}_{j=1} = \{y_{j}\}^{n}_{j=1} $, is $\hat{p}_{n} = \frac{\sum^{n}_{j=1}y_j}{n}$. By the law of large numbers,
\begin{align*}
\underset{n \to \infty}{\lim}  \frac{\sum^{n}_{j=1}Y_j}{n} = \underset{n \to \infty}{\lim} \hat{p}_{n} = p,\; \mathbb{P}\; \text{a.s}
\end{align*}
To quantify the rate of convergence, we invoke the LDP, (subsection \ref{Large Deviations Principle}). To invoke the same, a closed form expression of the joint distribution of $\{Y_{j}\}^{n}_{j=1}$ is needed. First, the joint distribution of $Y_{n}$ and $Y_{n+1}$ can be found as:
\begin{align*}
\mathbb{P}(Y_{n+1} =y_{n+1}, Y_{n} =y_{n}) &= \sum^{m}_{j=1}\mathbb{P}(Y_{n+1} =y_{n+1}, Y_{n} =y_{n},X_{n+1}= x_j)\\
&= \sum^{m}_{j=1}\mathbb{P}(Y_{n+1} =y_{n+1}| Y_{n} =y_{n},X_{n+1}= x_j) \mathbb{P}(Y_{n} =y_{n},X_{n+1}= x_j)\\
&= \Big(\mathbf{1}_{\{y_{n+1}=1\}}p + \mathbf{1}_{\{y_{n+1}=0\}}(1-p)\Big)\sum^{m}_{j=1}\mathbb{P}(Y_{n} =y_{n},X_{n+1}= x_j)\\
&= \Big(\mathbf{1}_{\{y_{n+1}=1\}}p + \mathbf{1}_{\{y_{n+1}=0\}}(1-p)\Big)\mathbb{P}(Y_{n} =y_{n})
\end{align*}
Then, by the principle of induction it follows that, 
\begin{align*}
\mathbb{P}(\{Y_{j} =y_{j}\}^{n}_{j=1}) = \prod^{n}_{j=1}\Big(\mathbf{1}_{\{y_{j}=1\}}p + \mathbf{1}_{\{y_{j}=0\}}(1-p)\Big)
\end{align*}
Since $\{Y_{j}\}^{n}_{j=1}$ is a sequence of i.i.d Bernoulli random variables, it follows that
\begin{align*}
\underset{n \to \infty}\lim \dfrac{1}{n} \ln \mathbb{P}_n\Big(\Big| \dfrac{\sum^{n}_{j=1}Y_j}{n} - p \Big| \geq \epsilon\Big) = \epsilon \ln \dfrac{p}{\epsilon} +   (1-\epsilon) \ln \dfrac{1-p}{1-\epsilon} =  -I(\epsilon) 
\end{align*}
\subsection{Modified Observation Model}
In observation model described in the previous section, the conditional distribution of $Y_{n}$ is independent of the current and past states of the agent, and, the previous observations collected by the agent. We define an alternative model where the observation at a given time instant depends on the current state of the system. The conditional distribution of $Y_{n}$ given $\{X_{j} = x_j\}^{n}_{j=1}, \{Y_{j} = y_j\}^{n}_{j=1} $ is modified to, 
\begin{align*}
&\bar{\pi}(Y_{n} = 1 | \{X_{j} = x_j\}^{n}_{j=1}, \{Y_{j} = y_j\}^{n}_{j=1}  ) = \bar{\pi}(Y_{n} = 1 | X_{n} = \bar{x}_j)  = p_j , \\
&\bar{\pi}(Y_{n} = 0 | \{X_{j} = x_j\}^{n}_{j=1}, \{Y_{j} = y_j\}^{n}_{j=1}  ) = \bar{\pi}(Y_{n} = 0 | X_{n} = \bar{x}_j)  = 1-p_j , 
\end{align*}
where $x_n = \bar{x}_j$ for some $j \in {1,\ldots, m}$. $\{p_j\}^{m}_{j=1}$ are unknown and are to estimated. Thus, the probability distribution at stage $n$, $\bar{\mathbb{P}}_{n}$,  changes (from $\mathbb{P}_n$) to 
\begin{align*}
\bar{\mathbb{P}}_{n}(\{x_j, y_j\}^n_{j=1}) = \prod^n_{j=1}\hspace{-4pt}\Big(\hspace{-2pt}\sum^{m}_{k=1}\big(\mathbf{1}_{\{x_{j} =\bar{x}_k, y_{j}=1\}}p_{k} + \mathbf{1}_{\{x_{j} =\bar{x}_k, y_{j}=0\}}(1-p_k)\big)\Big)\Big(\pi_{j-1}(x_{j-1},y_{j-1})[x_{j}] \hspace{-2pt}\Big).
\end{align*}
Assumption: The policy of the agent is the same at every time step, i.e., $\pi_{j}(\cdot)[\cdot] = \pi(\cdot)[\cdot],\; \forall j$. The, The above expression for join probability can be re-enumerated as follows. Consider any arbitrary sequence of observation, 
\begin{align*}
(x_{1}, y_{1}), \; (x_{2}, y_{2}),\; (x_{3}, y_{3}), \ldots, (x_{n-1}, y_{n-1}), \; (x_{n}, y_{n}) 
\end{align*} 
At every instant, $(X_{n}, Y_{n})$,  belongs to one of the $2m$ states, $(\bar{x}_{k}, y), k = 1, \ldots, m, y= 0,1$. From time instant $j$ to $j+1$, the state of the system can change from one of the $2m$ states to another of the $2m$ states. Thus, there are $2m \times 2m$ possible change of states. In an observation sequence of length $n$, there are $n-1$ change of states. Let $\eta(k,u,l,v)$ denote the number of transitions from the state $(\bar{x}_k,u)$ to the state  $(\bar{x}_l,v)$ where $k,l = 1, \ldots, m$ and $u,v=0,1$. Let $\pi^{l}_{k,u} :=  \pi(\bar{x}_k,u)[\bar{x}_l]$ and $p^{v}_{l} = \big(\mathbf{1}_{\{v=1\}}p_{l} + \mathbf{1}_{\{v=0\}}(1-p_l)\big) $. For $n \geq 2$, the joint probability $\bar{\mathbb{P}}_{n}$ can be expressed as:
\begin{align*}
\bar{\mathbb{P}}_{n}(\{x_j, y_j\}^n_{j=1}) = \Bigg[\prod^m_{k,l=1}\prod^m_{u,v=0,1} \Big(p^{v}_{l}\pi^{l}_{k,u} \Big)^{\eta(k,u,l,v)} \Bigg]\bar{\mathbb{P}}_{1}(x_{1}, y_{1}),\; \sum^{m}_{k,l=1}\sum_{u,v=0,1}\eta(k,u,l,v) = n-1. 
\end{align*}
Thus, the joint distribution of $\{Y_{j}\}^{n}_{j=1}$ is 
\begin{align*}
\mathbb{P}(\{Y_{j} =y_{j}\}^{n}_{j=1}) = \sum_{\{x_{j}\}^{n}_{j=1} \in \mathcal{X}^n} \Bigg[\prod^m_{k,l=1}\prod^m_{u,v=0,1} \Big(p^{v}_{l}\pi^{l}_{k,u} \Big)^{\eta(k,u,l,v)} \Bigg]\bar{\mathbb{P}}_{1}(x_{1}, y_{1}).
\end{align*}
The above expression does not yield itself immediately to the LDP. We introduce the following definition. 
\begin{definition}
(Asymptotic Independence)
The joint distribution of $\{Y_{j}\}^{n}_{j=1}$ is said to satisfy the asymptotic independence property, if there exists $\{\pi^{*}_{l}\}^{m}_{l=1} \in \{ \pi \in \mathbb{R}^{m}: \pi_{j}\geq 0, \sum_{j}\pi_{j}=1\}$ such that, 
\begin{align*}
 \sum_{\{x_{j}\}^{n}_{j=1} \in \mathcal{X}^n} \Bigg[\prod^m_{k,l=1}\prod_{u,v=0,1} \Big(p^{v}_{l}\pi^{l}_{k,u} \Big)^{\eta(k,u,l,v)} \Bigg]\bar{\mathbb{P}}_{1}(x_{1}, y_{1}) = \Big(\sum^{m}_{l=1}p_{l}\pi^{*}_{l}\Big)^{s}  \Big(\sum^{m}_{l=1}(1-p_{l})\pi^{*}_{l}\Big)^{n-s},
\end{align*}
for $n$ arbitrarily large, $s= \sum^{n}_{j=1}\mathbf{1}_{y_{j}=1}$.
\end{definition}
We note that $\pi^*$ corresponds to stationary distribution of the Markov chain formed by the states of the system . For the proof, we refer to subsection \ref{Discussion on Definition 1}.
\begin{proposition}
If $\{Y_{j}\}^{n}_{j=1}$ satisfies the asymptotic independence property, then 
\begin{align}
\underset{n \to \infty}\lim \dfrac{1}{n} \ln \bar{\mathbb{P}}_n\Big(\Big| \dfrac{\sum^{n}_{j=1}Y_j}{n} - \sum^{m}_{l=1}p_{l}\pi^{*}_{l} \Big| \geq \epsilon\Big) = \epsilon \ln \dfrac{\sum^{m}_{l=1}p_{l}\pi^{*}_{l}}{\epsilon} +   (1-\epsilon) \ln \dfrac{1-\sum^{m}_{l=1}p_{l}\pi^{*}_{l}}{1-\epsilon} \label{Equation 1}
\end{align}
\end{proposition}
For proof,  we refer to subsection \ref{Proof of Proposition 1}.
\subsection{Problem Definition}\label{Problem Definition}
We define the error in estimation as $e_{n} = \dfrac{\sum^{n}_{j=1}Y_j}{n} - p$. From equation \ref{Equation 1}, the probability of estimation error being greater than $\epsilon \in (0,1)$,  can be approximated as $\mathbb{P}_{n}(\epsilon > \epsilon) \approx e^{-nI(\epsilon)}$. Since $p$ is unknown, the estimated probability of the estimation error  being greater than $\epsilon$ is defined as, $\hat{\mathbb{P}}_{n}(\epsilon > \epsilon) \approx e^{-n \mathbb{E}_{\mathbb{P}_n}\big[I(\epsilon, \hat{p}_n)\big]}$, where $I(\epsilon, \hat{p}_n)$ is 
\begin{align*}
I(\epsilon, \hat{p}_{n}) = \epsilon \ln \dfrac{\epsilon}{\hat{p}_{n}} +   (1-\epsilon) \ln \dfrac{1-\epsilon}{1-\hat{p}_{n}}. 
\end{align*}
Due to the independence of the observations and states of the agent, resulting in i.i.d observations in the first observation model, the policy of the agent does not impact the observation sequence, and hence the estimate or the estimated probability of error. However in the second observation model, if we define $\bar{I}(\epsilon, \{p_{l,n}\}^{m}_{l=1})$ as, 
\begin{align*}
\bar{I}(\epsilon, \{p_{l,n}\}^{m}_{l=1}) = \epsilon \ln \dfrac{\epsilon}{\sum^{m}_{l=1}p_{l,n}\pi^{*}_{l}} +   (1-\epsilon) \ln \dfrac{1-\epsilon}{1-\sum^{m}_{l=1}p_{l,n}\pi^{*}_{l}},
\end{align*}
the estimated probability of error depends on the stationary distribution of the Markov chain formed by the states of the agent. Since the stationary distribution is dependent on the policy of the agent, the probability of error can be manipulated through policy of the agent. For a given decision making horizon, $N \in \mathbb{N}$, the objective is to find the stationary distribution, $\{\pi^*_{l}\}^{m}_{l=1}$, such that $\tilde{\mathbb{P}}_{N}(\epsilon > \epsilon)  \approx e^{-n\bar{I}(\epsilon, \{p_{l,N}\}^{m}_{l=})}$ is minimized for a fixed $\epsilon \in (0,1)$. Thus, the  decision making problem formulated as a optimization problem is, 
\begin{align*}
\underset{\{\pi^*_{l}\}^{m}_{l=1}}{\min} \; \mathbb{E}_{\bar{\mathbb{P}}_N} \Big[\epsilon \ln \dfrac{\sum^{m}_{l=1}p_{l,N}\pi^{*}_{l}}{\epsilon} +   (1-\epsilon) \ln \dfrac{1-\sum^{m}_{l=1}p_{l,N}\pi^{*}_{l}}{1-\epsilon}\Big].
\end{align*}
\section{Analysis of the Problem}\label{Analysis of the Problem}
In this section, we analyze the problem formulated above in a structured approach. We begin by studying the dependence  of the cost function  on the error bound $\epsilon$. We obtain an equivalent formulation of the problem and study it through a dynamic programming approach. 
\vspace{-0.5cm}
\subsection{Dependence on error bound}
%\begin{figure}
%\includegraphics[scale=0.6]{Figure 1_new.png}
%\caption{Rate function vs estimated probability}
%\label{Figure 1}
%\end{figure}
\begin{wrapfigure}{R}{0.5\textwidth}
\vspace{-0.8cm}
  \begin{center}
    \includegraphics[width=0.48\textwidth]{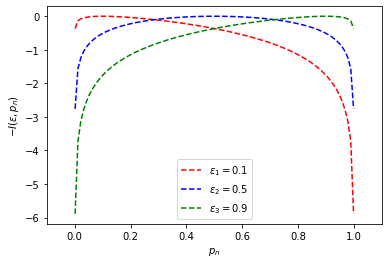}
  \end{center}
  \vspace{-1cm}
  \caption{Rate function vs estimated probability}
  \label{Figure 1}
\vspace{-0.5cm}
\end{wrapfigure}

In Figure \ref{Figure 1}, we have plotted $-I(\epsilon,\cdot)$ for three different values of $\epsilon$. For lower values $\epsilon$, higher rate (i.e., lower estimated probability of error) is achieved at higher values of estimated probability $p_n$, i.e., when $p_{n}$ is closer to $1$. Similarly, for higher values of $\epsilon$, higher rate is achieved at lower values of estimated probability $p_n$, i.e., when $p_n$ is loser to $0$. Hence, these choices of $\epsilon$ induce bias in the estimation problem, that is, the estimates get biased to $1$ or $0$, depending the choice of $\epsilon$. However, for $\epsilon=0.5$, we note that $-I(\epsilon,\cdot)$ is symmetric about $p_{n}=0.5$, i.e, the rate is the same at $p$ and $1-p$. Hence, for $\epsilon=0.5$, there is no bias. For $\epsilon=0.5$, the cost function gets modified to, 
\begin{align*}
\mathbb{E}_{\bar{\mathbb{P}}_N} [-\bar{I}(\frac{1}{2}, \{p_{l,N}\}^{m}_{l=1})] = \mathbb{E}_{\bar{\mathbb{P}}_N}\Bigg[\ln \sqrt{2\Big(\sum^{m}_{l=1}p_{l,N}\pi^{*}_{l}\Big)}  + \ln \sqrt{2\Big(1-\sum^{m}_{l=1}p_{l,N}\pi^{*}_{l}\Big)}\Bigg].
\end{align*}
By non-decreasing property of the $\ln(\cdot)$ and $\sqrt{(\cdot)}$ functions, and, the monotonicity of expectation, it suffices to minimize,
\begin{align*}
\mathbb{E}_{\bar{\mathbb{P}}_N} \Big[\Big(\sum^{m}_{l=1}p_{l,N}\pi^{*}_{l}\Big)\Big(1-\sum^{m}_{l=1}p_{l,N}\pi^{*}_{l}\Big)\Big].
\end{align*}
Thus, the negative rate depends on the product of the estimated probability that  $Y_n=1$ and the estimated probability that $Y_n=0$. 
\vspace{-0.5cm}
\subsection{Dynamic Programming Approach}
Suppose the above estimation cost and problem can be expressed of the form, 
\begin{align*}
\underset{\{\pi_{j}(X_{j}, Y{j})[\cdot]\}^{n-1}_{j=1}} {\min} \mathbb{E}_{\bar{\mathbb{P}}_N} \Big[\Big(\dfrac{\sum^{N}_{j=1}Y_{j}(\pi_{j-1}(X_{j-1}, Y_{j-1})[\cdot])}{N}\Big)\Big(1- \dfrac{\sum^{N}_{j=1}Y_{j}(\pi_{j-1}(X_{j-1}, Y_{j-1})[\cdot]))}{N}\Big)\Big],
\end{align*}
where $Y_{j}(\pi_{j-1}(X_{j-1}, Y_{j-1})[\cdot])$ denotes the observation at stage $j$ due to the policy \\ $\pi_{j-1}((X_{j-1}, Y_{j-1})  [\cdot])$ at stage $j-1$. The cost function can be  split in to stage costs, 
\begin{align*}
&S_{N+1}\Big(\{Y_j\}^{N}_{j=1}\Big) = 0, \;S_{N}\Big(\{Y_j\}^{N}_{j=1}\Big) =   \Big[N Y_{N} - Y^{2}_{N}-2\sum^{N-1}_{j=1}Y_{N}Y_{j}\Big],\\
&S_{n}\Big(\{Y_j\}^{i}_{j=1}\Big) = \Big[N Y_{i} - Y^{2}_{i}-2\sum^{i-1}_{j=1}Y_{i}Y_{j}\Big], 2 \leq n \leq N-2, \; S_{1}(Y_1)  =  [ N Y_{1}- Y^2_{1}].
\end{align*}
The Bellman's Equation for the above problem is:
\begin{align*}
&V_{n}(X_{n},\{Y_j\}^{n}_{j=1}) = \underset{\pi_{n}(\cdot)[\cdot]}{\min} \mathbb{E}_{\bar{\mathbb{P}}_N}\Bigg[ S_{n}\Big(\{Y_j\}^{n}_{j=1}\Big)  + \\
&\hspace{3.85cm}\sum_{y=0,1}\Big[V_{n+1}\Big(\{Y_j\}^{n}_{j=1};y)\Big)\Big(\sum^{m}_{l=1}p^{y}_{l}\pi_{n}(X_{n}, Y_{n})[\bar{x}_{l}]\Big) \Big]\Bigg| \sigma \Big( X_{n}, \{Y_{j}\}^{n}_{j=1}\Big)\Bigg], \\
&V_{n}(X_{n},\{Y_j\}^{n}_{j=1}) =S_{n}\Big(\{Y_j\}^{i}_{j=1}\Big)  +  \underset{\pi_{n}(\cdot)[\cdot]}{\min}   \sum_{y=0,1}\Big[V_{n+1}\Big(\{Y_j\}^{n}_{j=1};y)\Big)\Big(\sum^{m}_{l=1}p^{y}_{l,n}\pi_{i}(X_{n}, Y_{n})[\bar{x}_{l}]\Big)\Big]
\end{align*}
where $1\leq i \leq N-1$ and $V_{N+1}(\{Y_j\}^{N}_{j=1}) =0$. We make three observations about the above equation. (1) At a given stage, $n$, the stage cost is not impacted by the policy at the stage. Only the cost to go of the the next stage is influenced by the current policy. This is to be expected as there are costs involving the ``norm" of the policy. (2) In the first equation the $\{p^{y}_{l}\}^{m}_{l=1}$ is unknown. In the following equation, we invoke the separation principle \citep{kumar2015stochastic} allowing us to use an estimate of these probabilities. (3) The curse of dimensionality. The state space of the value function changes with every iteration. For large $N$, we could have to compute the value function for $\mathcal{X} \times \{0,1\}^{N}$ values, i.e., the domain of the state space growing exponentially. Though the value functions can be computed offline, this approach is not a computationally effective to  active learning 
\vspace{-0.7cm}
\subsection{Relaxation of The Problem}
Given the challenges posed by the dynamic programming approach, we consider a relaxation of the problem formulated in subsection \ref{Problem Definition}. Consider the optimization problem, 
\begin{align*}
\underset{\{\pi^{*}_l\}^{m}_{j=1}}{\min}  \mathbb{E}_{\bar{\mathbb{P}}_N}\Big[ \Big(\sum^{m}_{l=1}p_{l}\pi^{*}_l\Big) \Big(1- \sum^{m}_{l=1}p_{l}\pi^{*}_l\Big) \Big| \sigma \Big( X_{i}, \{Y_{j}\}^{i}_{j=1}\Big)\Big] 
\end{align*} 
Invoking the separation principle, at iteration $N$, the following optimization problem is solved. 
\begin{align}
\underset{\{\pi^{*}_l\}^{m}_{j=1}}{\min} \Big(\sum^{m}_{l=1}p_{l,N}\pi^{*}_l\Big) \Big(1- \sum^{m}_{l=1}p_{l,N}\pi^{*}_l\Big),\; \text{s.t}\; \sum_{l}\pi^{*}_l= 1. \label{Equation 2}
\end{align} 
Since the cost function is concave in the optimization variables, the optimizer is one of the vertices of the probability simplex \citep{rockafellar1970convex}, i.e., $\pi^{*}= [0,\ldots, 1, \ldots, 0]$, where $\pi^*_{l^*} =1$ if $l^*$ is such that 
\begin{align*}
\Big(p_{l^*,N}\Big) \Big(1- p_{l^*,N}\Big) \leq \Big(p_{l,N}\Big) \Big(1- p_{l,N}\Big), l \in \{1, \ldots, m\}.
\end{align*}
Using the stationary distribution of the $\{X_n\}$  Markov chain, we can retrieve the stationary distribution of $\{X_n,Y_n\}$ Markov chain as, $\bar{\pi}^{*} = [0,\ldots, p_{l^*,N}, \ldots, 0, 0, \ldots, 1-p_{l^*,N}, \ldots,0]$. Given the stationary distribution, retrieving the policy $\pi(\cdot)[\cdot]$ is an ill-posed problem as there are $2m \times m$ variables and only $2m$ equations. Solving the equations lead to the following solution,
\begin{align*}
\pi^{*}(\bar{x}_{l^*}, 1)[\bar{x}_{l^*}] =\dfrac{p_{l^*,N}}{p^2_{l^*,N} + (1-p_{l^*,N})^2}, \pi^{*}(\bar{x}_{l^*}, 0)[\bar{x}_{l^*}] =\dfrac{1-p_{l^*,N}}{p^2_{l^*,N} + (1-p_{l^*,N})^2},  
\end{align*}
and, $\pi^{*}(\bar{x}_{k}, y)[\bar{x}_{l^*}] = 0, \;k \in \{1, \ldots, m\}, \; k\neq l^*, y=0,1$. Since there are no conditions for the remaining $2m \times (m-1)$ variables, the policy $\pi^{*}(\bar{x}_{k}, y)[\bar{x}_{l}], \;l,k \in \{1, \ldots, m\}, \; l, k \neq l^*, y=0,1$ can be chosen at random from the interval $[0,1]$. Though this policy can be executed in practice, we are unable to interpret the policy. Retaining the essence of the optimization problem in equation \ref{Equation 2}, given the current state ($X_{n}=\bar{x}_{k}, Y_{N}=y$), we define a sequence of optimization problems, 
\begin{align}
(P_n): \underset{\{\pi_{n}(\bar{x}_k,y)[\bar{x}_l]\}^{m}_{l=1}}{\min} &\Big(\sum^{m}_{l=1}p_{l,n}\pi_{,}(\bar{x}_k,y)[\bar{x}_l]\Big) \Big(1- \sum^{m}_{l=1}p_{l,n}\pi_{n}(\bar{x}_k,y)[\bar{x}_l]\Big), \nonumber \\ 
 \text{s.t}\; &\sum_{l}\pi_{n}(\bar{x}_k,y)[\bar{x}_l] = 1, n \in \mathbb{N} \label{Equation 3}
\end{align} 
where $\pi_{n}(\bar{x}_k,u)[\bar{x}_l]$ denotes the probability of transition from state $X_{n}=\bar{x}_{k}, Y_{n}=u$ to state $X_{n+1}=\bar{x}_{l}$ as defined in subsection \ref{Problem Setup}.
\vspace{-0.5cm}
\section{Algorithm}\label{Algorithm}
In this section, we present the active parameter learning algorithm and the classification algorithm to identify the safe regions
\vspace{-0.5cm}
\subsection{Parameter Learning Algorithm}
Similar to the optimization problem in \ref{Equation 2}, the cost function is concave in the optimization variables in  the problem in \ref{Equation 3}. Thus, pure strategies are optimal if the minimum cost is achieved at a unique vertex of the probability simplex. The Markov chain transitions from the state ($X_{n}=\bar{x}_{k}, Y_{n}=u$) to the state $X_{n+1}=\bar{x}_{l^*}$ with probability $1$, i.e., $\pi^{*}_{n}(\bar{x}_k,y)[\bar{x}_{l^*}] =1$  if there exits unique ${l^*}$ is such that, 
\begin{align}
\Big(p_{l^*,n}\Big) \Big(1- p_{l^*,n}\Big) \leq \Big(p_{l,n}\Big) \Big(1- p_{l,n}\Big), l \in \{1, \ldots, m\}. \label{Equation 4}
\end{align}
If there are multiple states with the same minimum cost, then one of them is chosen at random, with a uniform distribution. That is, if the states $\{\bar{x}_{l_{1},n},\ldots, \bar{x}_{l_{r}, n} \}$ have the same minimum cost, then $\pi^{*}_{n}(\bar{x}_k,y)[\bar{x}_{l_{u}, n}] = \frac{1}{r}, u = 1, \ldots, r$. 

\noindent The above policy introduces bias in the following sense. After $N$ iterations, suppose the state $\bar{x}_{k}$ is the unique state satisfying the  inequality  in \ref{Equation 4} and $N$ is such that
\begin{align*}
|p_{k, n} - p_{k, n+1}| < \underset{l \in \{1, \ldots, m\}, l \neq k} {\inf} \Big(p_{l,N}\Big) \Big(1- p_{l,N}\Big) - \Big(p_{k,N}\Big) \Big(1- p_{k,N}\Big), \forall n \geq N.
\end{align*} 
Then, the optimal policy at every iteration is to stay at stage $k$, for all $n \geq N$. In order to overcome this bias, we introduce parameters $N_{\max}$, $N_{\delta}$ and $\delta$. If the optimal policy is to stay at state $k$ for more than $N$ iterations, i.e., $\pi_{j}(\bar{x}_k,u)(\bar{x}_k) = 1$ for $j \in [n- N_{\max}, n]$ or $|p_{k, j} - p_{k, n} | < \delta, \forall j \in [n- N_{\delta}, n]$, then $p_{k,n}$ is the estimate at state $k$ and the state is removed from further computations. The optimization problem is now executed with the states $\{\bar{x}_1, \ldots,\bar{x}_m\}\sim \{\bar{x}_k\}$. The pseudo code corresponding to the algorithm is stated in Algorithm \ref{Algorithm 1}. 
\begin{algorithm}
\caption{Active Parameter Learning Algorithm}
\begin{algorithmic}[1]\label{Algorithm 1}
\STATE Initialize $\delta, N_{\delta},N_{\max}$, $k_{1} \sim \text{unif}\{1,m\}$, $n \gets 1$, $X_{1} \gets \bar{x}_{k_1}$, $current-state \gets \bar{x}_{k_1}$
\STATE $past-state \gets \bar{x}_{k_1} $, $active-states\gets \{\bar{x}_{1}, \ldots, \bar{x}_{m}\}$, $counter \gets zero$, $max \gets 0$
\WHILE {$active-states \neq \varnothing$}
\STATE Collect observation $Y_{n}$ for current iteration.
\STATE Update $p_{n,k}$ where $k$ is index corresponding to the $current-state$
%\IF    {$ past-state = current-state$}
%\STATE $counter \gets counter +1$
%\ELSE 
%\STATE $counter \gets 0$
%\ENDIF
%\IF    {$n \geq N_{\delta}$}
%\STATE $j \gets 1$
%\WHILE {$j \leq N_{\delta}$}
%\STATE $\text{temp} \; = \; | p_{k,n - N_{\delta}+j} - p_{k,n -N_{\delta}} |$
%\IF    {$ \text{temp} > \max$}
%\STATE $\max \gets \text{temp}$
%\ENDIF
%\STATE $j \gets j+1$
%\ENDWHILE
%\ENDIF
\STATE If $\pi_{j}(\bar{x}_k,u)(\bar{x}_k) = 1$ for $j \in [n- N_{\max}, n]$, $Condition_1 = TRUE$
\STATE If $|p_{k, j} - p_{k, n} | < \delta, \forall j \in [n- N_{\delta}, n]$, $Condition_2 = TRUE$
%\IF    {$ counter = N_{\max}$ or $ temp < \delta$}
\IF    {$Condition_1 = TRUE$ or $Condition_2 = TRUE$}
\STATE $active-states \gets  active-states \sim \{\bar{x}_k\} $, 
\STATE $c_{k} \gets (p_{k,n})(1-p_{k,n})$
\ENDIF
\IF    {$active-states \neq \varnothing$}
\STATE Find policy at current iteration, $\pi^{*}_{n}(\bar{x}_k,y)[\cdot]$
\STATE $past-state  = current-state$, $n\gets n+1$
\STATE Execute policy and move to next state
\ENDIF
\ENDWHILE
\end{algorithmic}
\end{algorithm}
\vspace{-0.5cm}
\subsection{Robust Classification Algorithm}
The objective is to define a classification rule to label the cells safe or unsafe. The classification rules are typically threshold based. First, we note that defining safety as a function of $p_{k,n}$ is not suitable. This is because,  when $p_{k,n}$ is greater than $\frac{1}{2}$ but close to $\frac{1}{2}$, after few iterations it might be less than $\frac{1}{2}$, i.e., it is not a robust measure. However, $(p_{k,n})(1-p_{k,n})$ is a robust measure. When $p_{k,n}$ is greater than $\frac{1}{2}$  and the measure is low, it implies $p_{k,n}$ is closer $1$ which implies that the region is safe. When the measure is high, it means that the agent is not able to distinguish between $p_{k,n}$ and $1-p_{k,n}$ and such a region could be marked unsafe. 

Formalizing the same, we define $c_{k}  = (p_{k,n_k})(1-p_{k,n_k})$, where $n_{k}$ is the last iteration after which the state $\bar{x}_{k}$ is eliminated from the optimization. We could consider absolute classification, where we define a threshold $c_{T}$ independent of $\{c_{k}\}^{m}_{k=1}$. If $c_{k} \leq c_{T}$ and $p_{k,n_k} \geq \frac{1}{2}$, then the region is said to be safe and unsafe otherwise. We consider relative classification, i.e, $c_{T}$  is a function of  $\{c_{k}\}^{m}_{k=1}$.   Let $c_{mean}$ and $c_{median}$ be the mean and median of $\{c_{k}\}^{m}_{k=1}$. W.L.O.G, let $c_{1} \leq c_{2} \leq \ldots c_{m-1} \leq c_{m}$. We define the threshold as follows. 
\begin{definition}
Suppose $k^*$ is such that $c_{k^*+1} -c_{k^*}$ is greater than or equal to  $c_{k+1} -c_{k}$ for all other $k$. If $k^* \in [\lceil \frac{m}{4}\rceil, \lceil \frac{3m}{4}\rceil] $, then $c_{T} = c_{k^*}$. Otherwise, $c_{T} = c_{median}$ 
\end{definition}
The above definition of the threshold ensures that on an average, appropriately half the regions are marked safe. The above definition can be explained as follows. If there is skewness in the $\{c_{k}\}^m_{k=1}$ values, that is, there is a partition of the collection into two subsets where one subset can be differentiated from the other through a``large" difference in $c_{k}$ values, then the threshold corresponds to the lower value in the differentiation. Otherwise, if there is no skewness the threshold corresponds to $c_{median}$.  
\begin{definition}
The Voronoi Cell with center $\bar{x}_{k}$ is said to be safe if $c_{k} \leq c_{T}$ and $p_{k,n_k} \geq \frac{1}{2}$. Let $ \mathcal{X}_s$ denote the set of Voronoi centers corresponding to the safe sets. 
\end{definition}
\begin{proposition}\label{Proposition 5}
$\exists N_{\omega}$ such that $X_{n}(\omega) \in \mathcal{X}_s, \; \forall n \geq N_{\omega}, \bar{\mathbb{P}} \; a.s$.
\end{proposition}  
For the pseudo code of the classification algorithm and proof of the above proposition, we refer to subsection \ref{Safety Algorithm}. 
\section{Numerical Example}\label{Example}
We consider the example described in Figure \ref{Figure 2}. The simulation set up is described as follows. We let $\delta=0.02$, $N_{\delta} = 250$, $N_{\max} = 50$. The true probability of observing $1$ given the center is listed in the first row of  Table \ref{Table 1}. With this setup, the Algorithms were executed. The estimated  probability of observing $1$ given the Voronoi center after 350 iterations is listed in the second row of  Table \ref{Table 1}. At this juncture it becomes clear that regions corresponding to centers $\bar{x}_{1}, \bar{x}_{2}, \bar{x}_{4}, \bar{x}_{6}, \bar{x}_{8}$ are safe while regions corresponding to $\bar{x}_{3},\bar{x}_{5}, \bar{x}_{7},\bar{x}_{9}$ are unsafe. However, the parameter learning algorithm is executed further. The final estimates of the probabilities are listed in the third row of Table \ref{Table 1}. The order in which the safe regions were identified is depicted in figure \ref{Figure 3}. Safe regions are depicted in green while unsafe regions are depicted in red. It was observed that that the MDP  stayed in region $1$,followed by region $2$, region $4$, then region $6$ and finally stopped at region $8$. At each region the probability of observing $1$ was estimated to the accuracy of $0.02$. The parameter learning algorithm is stopped at this juncture. Using the estimated probabilities, $c_{3}-c_{8}$  was found to be maximum. Hence, $c_{T}=c_{8}=0.1936$ which also equals $c_{median}$. Using this threshold value, the regions are labeled safe or unsafe. The same has been depicted on the right side of Figure \ref{Figure 3}. 
\begin{table}[ht!]
\vspace{-0.2cm}
\begin{center}
\begin{tabular}{|| c | c | c | c | c | c | c | c | c | c ||} 
\hline
$k$  &   $\bar{x}_1$  &  $\bar{x}_2$  & $\bar{x}_3$   & $\bar{x}_4$ &$\bar{x}_5$  &  $\bar{x}_6$  & $\bar{x}_7$   & $\bar{x}_8$ &$\bar{x}_9$ \\
\hline
$p_{k}$ & $0.9$ & $0.84$ & $0.58$ & $0.79$ & $0.6$ & $0.75$ & $0.54$ & $0.72$ & $0.56$ \\
\hline
$p_{k,n}$ & $0.875$ & $0.80$ & $0.625$ & $0.775$ & $0.525$ & $0.725$ & $0.575$ & $0.675$ & $0.475$ \\
\hline
$p_{k, n_{k}}$ & $0.91$ & $0.85$ & $0.625$ & $0.80$ & $0.525$ & $0.7625$ & $0.575$ & $0.7375$ & $0.475$ \\
\hline
\end{tabular}\\
\vspace{-0.2cm}
\caption{Estimates of the parameters at different iterations}
\label{Table 1}
\end{center}
\vspace{-1cm}
\end{table} 
\begin{figure}
\begin{center}
\includegraphics[scale=0.38]{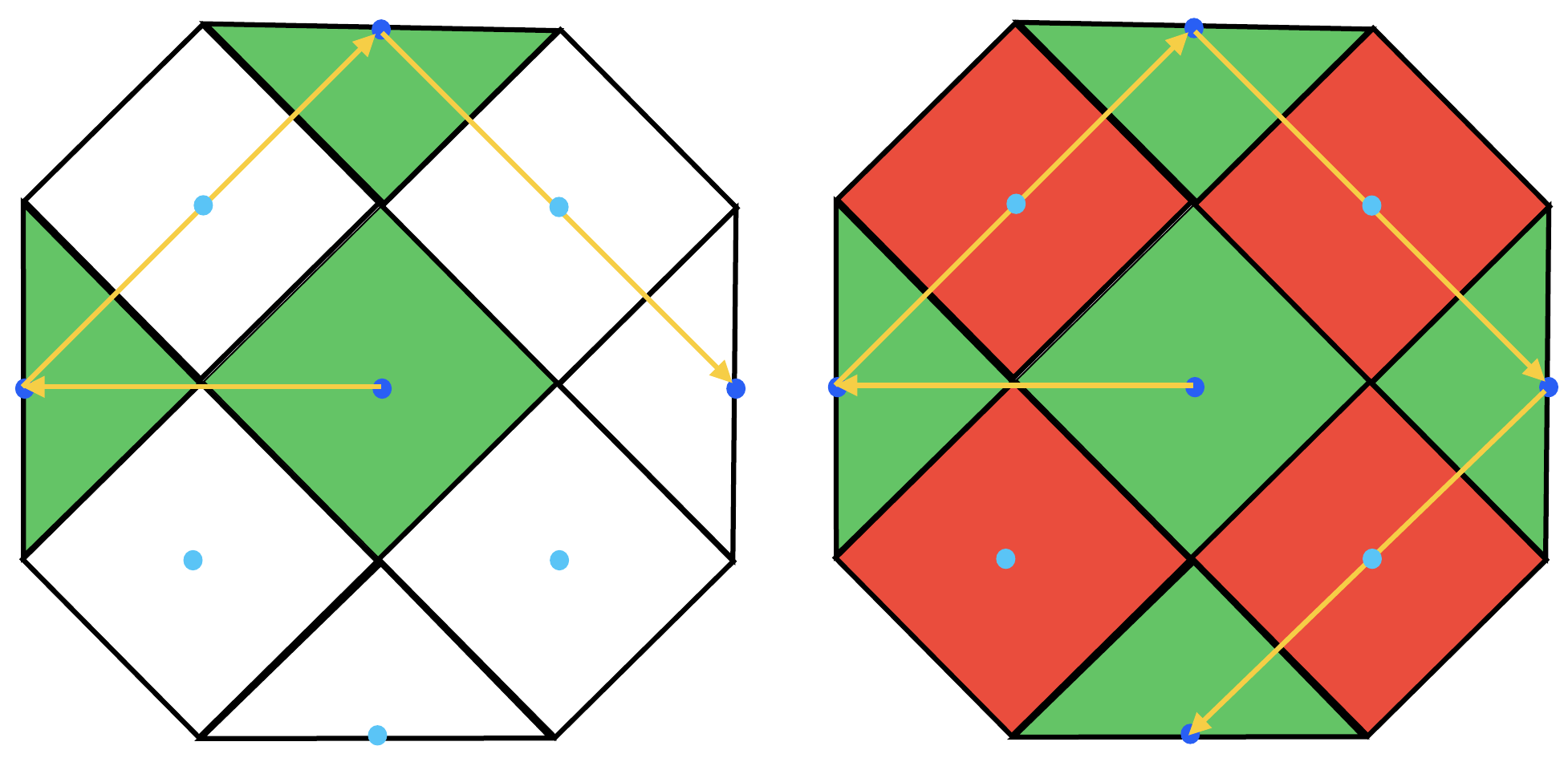}
\end{center}
\caption{(Left) 3 safe regions identified and transit to fourth region. (Right) All safe and unsafe regions identified by the agent}
\label{Figure 3}
\end{figure}

\section{Conclusions and Future Work}\label{Conclusion and Future Work}
Thus, in this paper we considered a control theoretic formulation of an active learning problem. Due to the computational challenges posed by the dynamic programming equations, we formulated a relaxed one step optimization problem. This optimization problem can be interpreted as model predictive control problem with prediction horizon of $1$ time step. Using the solution of the optimization problem, we proposed a active parameter learning algorithm and a robust classification algorithm. As future work, we are interested in characterization of the rate of convergence of the prescribed learning algorithm and compare it with the LDP. We are also interested in studying  the problem for more general observation models for e.g. the exponential family of distributions.  
\acks{This work was supported in part by Swedish Research Council Distinguished Professor Grant 2017-01078, Knut and Alice Wallenberg Foundation Wallenberg Scholar Grant, and the Swedish Strategic Research Foundation FUSS SUCCESS Grant.}
\bibliography{biblio}
\section{Appendix}
\subsection{Topological Vector Space of Real valued Measures}\label{Vector Space of Measures}
Let $\mathfrak{M}(\mathcal{F}, \mathbb{R})$ denote the vector space of of real valued measures on a measurable space, $(\Omega, \mathcal{F})$, where $\big(\Omega, \mathcal{T}(\Omega)\big)$ is a topological space. A topology is needed on $\Omega$ to define continuous functions on the same. There are different topologies that can be associated with this vector space. We define three of them here. 
\begin{definition}
Let $\{\mu_{n}\}_{n \geq 1} \subset \mathfrak{M}(\mathcal{F}, \mathbb{R})$ be a sequence of measures. 
\begin{itemize}
\item $\{\mu_{n}\}_{n \geq 1}$  converges uniformly / in norm / in total variation to $\mu$  if
\begin{align*}
\underset{n \to \infty}{\lim} \;\; \underset{E \in \mathcal{F}}{\sup}\;\; |\mu_{n}(E) - \mu(E)| = 0.
\end{align*}
\item $\{\mu_{n}\}_{n \geq 1}$  converges strongly / pointwise to $\mu$  if $\underset{n \to \infty}{\lim} \;\; \mu_{n}(E) = \mu(E), \forall E \in \mathcal{F}$. 
\item $\{\mu_{n}\}_{n \geq 1}$  converges weakly / in dual sense to $\mu$  if $\underset{n \to \infty}{\lim} \;\; \int_{E}f d\mu_{n} = \int_{E}f d\mu, \forall E \in \mathcal{F}, \; \forall f : \Omega \to \mathbb{R}$, continuous and bounded.  
\end{itemize}
\end{definition}
Let $\mathscr{P}(F) \subset \mathfrak{M}(\mathcal{F}, \mathbb{R})$ denote the set of probability measures on $(\Omega, \mathcal{F})$, i.e., $\mathscr{P}(F) =\{\mu \in \mathfrak{M}(\mathcal{F}, \mathbb{R}): \mu(\varnothing) =0, \mu(\Omega) = 1\}$. When $\Omega = \mathbb{R}^d$ and $\mathcal{F} = \mathcal{B}(\mathbb{R}^d)$, and, $\mu \in  \mathscr{P}(\mathcal{B}(\mathbb{R}^d))$ is absolutely continuous with respect to the Lebesgue measure, $\Lambda$ on $(\mathbb{R}^{d}, \mathcal{B}(\mathbb{R}^d))$, the Radon-Nikodym derivative of $\mu$ with respect to $\Lambda$ is referred to as the density of $\mu$, i.e., $\exists \dfrac{d\mu}{d\Lambda}: \Omega \to \mathbb{R}$, measurable such that, $\mu(E) = \int_{E} \dfrac{d\mu}{d\Lambda}(\omega)d\Lambda(\omega), \forall E \in \mathcal{F}$. 
\subsection{Large Deviations Principle}\label{Large Deviations Principle}
Let $\mathcal{E}$ be an experiment with random outcomes. The sample space for the experiment $\mathcal{E}$, denoted by $\Omega_{\mathcal{E}}$, is the set of all possible outcomes of the experiment. A probability space is defined as a triple $(\Omega, \mathcal{F}, \mathbb{P})$ where: $\Omega$ is a set, $\mathcal{F}$ is a sigma algebra of the subsets of $\Omega$, and $\mathbb{P}$ is a probability measure on $\mathcal{F}$. A random variable on this probability space is measurable function, $f: \Omega \to \mathcal{X}$, where $\mathcal{X}\subset \mathbb{R}^d$.  Usually when it is stated that $(\Omega, \mathcal{F}, \mathbb{P})$ is a probability space, there is no reference to an experiment, it is implicitly understood that $\Omega$ corresponds to the sample space of some experiment. Since the experiment is not explicitly known, $\omega \in \Omega$ has no interpretation. This definition of the probability space is abstract and hence it is not possible to impose additional assumptions on $\Omega$ without knowing the actual experiment, specifically on the topology of  $\Omega$ even if there exists one.

An alternative way to define a probability space would be to first define the experiment. The sample space of the experiment comprises of all the outcomes of the experiment. By suitably defining random variables or observables,  the set of all measurable outcomes of the experiment is obtained, which is denoted by $\Omega$. The set of all events, $\mathcal{F}$, is defined as the collection of all (open) subsets of $\Omega$. The distribution of the events, $\mathbb{P}$, is found by repeating the experiment an arbitrarily large number of times and the calculating the relative frequency of the events. Then the probability space is $(\Omega, \mathcal{F}, \mathbb{P})$. This approach to construct a probability space is followed while solving problems in statistics, especially to analyze the properties of the estimator. In such a scenario, it is reasonable to  make assumptions on the topology of $\Omega$ as long as the outcomes of the particular experiment are compatible with the assumptions being made. The definition of the large deviations principle as mentioned below requires additional assumptions on the sample space described as follows.

Let $(\Omega, d_{\Omega})$ be a polish space, i.e., complete and separable metric space and $\mathcal{F}$ be the Borel sigma algebra of $\Omega$. Let $\mu$ be a default measure on $(\Omega, \mathcal{F})$, for e.g. the Lebesgue measure on $(\mathbb{R}^{d}, \mathcal{B}(\mathbb{R}^d))$. Let $\{\mathcal{P}_{n}\}_{n\geq 1}$ be a sequence of probability measures on $\mathcal{F}$. Typically, the sequence of measures converges weakly to a measure $\mathcal{P}$ which is degenerate, i.e., $\dfrac{d\mathcal{P}}{d\mu} =0$, almost every where. For most of the sets in the sigma algebra, $\mathcal{P}_{n}(A)$ weakly converges to zero. The objective is to characterize the rate of convergence.
\begin{definition}\citep{varadhan2016large}
The sequence of measures $\{\mathcal{P}_{n}\}_{n\geq 1}$ is said to satisfy the large deviations principle (LDP) with rate function $I$, if there exists a  lower semicontinuous function $I:  \to [0,\infty]$ satisfying,
\begin{itemize}
\item for each $l$, the set  $\{x: \mathbb{I}(x) \leq l\}$ is compact,
\item for every closed subset $\bar{E}$ of $\Omega$,  
\begin{align*}
\underset{n \to \infty }\limsup \;\frac{\log \mathcal{P}_{n}(\bar{E})}{n} \leq  -\underset{x \in \bar{E}}{\inf}\; I(x),
\end{align*}
\item for every open subset $O$ of $\Omega$,  
\begin{align*}
\underset{n \to \infty }\liminf \; \frac{\log \mathcal{P}_{n}(O)}{n} \geq  -\underset{x \in O}{\inf}\; I(x).
\end{align*}
\end{itemize}
\end{definition}
Example: let $\{Y_{n}\}$ be a sequence of i.i.d random variables drawn from  a Bernoulli distribution with success probability $p$, i.e, $\mathbb{P}(Y_{n} = 1) = p$ and  $\mathbb{P}(Y_{n} = 0) = 1-p$. Then, $S_{n} = \dfrac{\sum^n_{j=1}Y_j}{n}$ satisfies the LDP with the rate function, 
\begin{align*}
I(x) = x \ln \dfrac{x}{p} +   (1-x) \ln \dfrac{1-x}{1-p}, x \in (0,1).  
\end{align*}
\subsection{Concentration of Measure}\label{Concentration of Measure}
Consider the cone of convex Lipschitz real valued functions on $(\Omega, d_{\Omega})$, i.e., 
\begin{align*}
E_{C^{k,1}} =\Big\{ f: \Omega \to \mathbb{R}: f(t\omega_1 + (1-t)\omega_2) \leq tf(\omega_1) + (1-t)f(\omega_2), \forall t\in [0,1] \\
|f(\omega_1) - f(\omega_2)| \leq c_{f}d_{\Omega}(\omega_1, \omega_2) \Big\}
\end{align*}
Let $\mathbb{M}_{\mathcal{P}_n}[f]$ denote the median of $f$ with respect to $\mathcal{P}_n$.
\begin{definition}\citep{talagrand1988isoperimetric,meckes2012concentration}
\end{definition}
The sequence of measures $\{\mathcal{P}_n\}$  is said to satisfy the the subgaussian convex concentration property (CCP), if there exists sequences , $\{c_{1}(n)\}_{n \geq 1}, \{c_{2}(n)\}_{n \geq 1} \subset \mathbb{R}_{++}$, such that, 
\begin{align*}
\mathcal{P}_{n}(| f - \mathbb{M}_{\mathcal{P}_n}[f] | \geq \epsilon ) \leq c_{1}(n)e^{-\frac{c_{2}(n)}{c_f}\epsilon^2}, \forall f \in E_{C^{k,1}}
\end{align*}
\subsection{Classification Problem}\label{Classification Problem}
Given a set $\mathcal{X} \subset \mathbb{R}^d$, every point in the set is labeled either $1$ or $0$, i.e., there exists a function $f$, such that $f(x) \in \{1,0\}$ is the true label of the point $x \in \mathcal{X}$. Usually $f(\cdot)$ is unknown. Suppose $X$ is drawn at random from $\mathcal{X}$ and the joint distribution between $(X, f(X))$, $\mathbb{P}$, is known. Then to estimate $f(\cdot)$, or the most likely label of $X$,  the following minimization problem is considered. 
\begin{align*}
\underset{f \in L^{1}(\mathcal{X} \times \{-1,0\}, \mathcal{B}(\mathcal{X}), \mathbb{P})}{\min} \; C(f), \; C(f) = \mathbb{E}_{\mathbb{P}}[Y(1-f(X)) + (1-Y)f(X)]
\end{align*}
The classification rule is given by:
\begin{align*}
f^*(x) = \left\{ \begin{aligned} 
  1, \; \text{if} \; \mathbb{P}(Y=1| X=x) \geq 0.5\\
  0, \; \text{if} \; \mathbb{P}(Y=1| X=x) < 0.5
\end{aligned} \right.
\end{align*} 
When the conditional distribution is known and $Y$ is not observed the above classification law is implemented. Suppose $Y$ is also observed, that is, given a finite sequence of observations $\{X_{j}, Y_{j}\}^n_{j \geq 1}$, where $X_{n}(\omega) \in \mathbb{R}^d$, $Y_{n}(\omega) \in \{-1,1\}, \; \forall n$, are i.i.d, with known joint distribution, a model estimation problem can be formulated as follows, 
\begin{align*}
\underset{f \in \mathfrak{F}}{\min}\; \bar{C}(f), \; \bar{C}(f) = \mathbb{E}_{\mathbb{P}}\Big[\sum^{n}_{j=1}Y_{j}(1 -f(X_{j})) + f(X_{j})(1 - Y_{j}) \Big],
\end{align*}
where the function space $\mathfrak{F}$ is to be chosen suitably. When $\mathfrak{F}$ is an RKHS the above problem can be solved using the stochastic representer theorem, \cite{raghavan2023distributed}. When the joint distribution is unknown the expected cost  is replaced by the empirical cost and the optimization problem changes to 
\begin{align*}
\underset{f \in \mathfrak{F}}{\min}\; C_{emp}(f), \; C_{emp}(f) =\frac{1}{n} \sum^{n}_{k=1}\Big[ Y_k(1-f(X_k)) + (1-Y_k)f(X_k) \Big]. 
\end{align*}
In the literature, for e.g. \cite{von2011statistical}, it has been shown using the union bound and the  Hoeffding’s Lemma that, 
\begin{align*}
\mathbb{P}(|C_{emp}(f^*_n) - \bar{C}(f^*)| > \epsilon ) \leq   \mathbb{P}(\underset{f \in \mathcal{F}}{\sup} \;|C_{emp}(f) - C(f)| > \epsilon ) \leq \mathcal{N}(\mathfrak{F},2n)  e^{\frac{-n\epsilon^2}{4}},
\end{align*}
where $\mathcal{N}(\mathfrak{F},2n)$ is the shattering coefficient. The above expression characterizes the rate of decay of the estimation error. To control the rate of decay, one should be able to control the upper bound. Since $\{X_{n}, Y_{n}\}$ is an i.i.d sequence, we note that the R.H.S does not posses any control parameter. Further investigation is needed to characterize the error rate when $\{X_{n}, Y_{n}\}$ is a Markov chain using concentration inequalities. In the following subsection, we use the large deviations principle to bound the estimation error. 
\subsection{Proof of Proposition 1}\label{Proof of Proposition 1}
The standard approach to prove LDP for Bernoulli r.v.s is to use the \textit{Chernoff's} Inequality. Let $\{Y_{n}\}$ be a sequence of i.i.d Bernoulli r.v.s  with parameter $p$. Then, for all $\theta > 0$,
\begin{align*}
\mathbb{P}_{n} \Big(\sum^{n}_{j=1}Y_{j}\geq n\epsilon\Big) \leq \dfrac{e^{\theta \big(\sum^{n}_{j=1}Y_{j}\big)}}{e^{\theta n \epsilon}}= \dfrac{(e^{\theta}p + (1-p))^{n}}{e^{\theta n \epsilon}},
\end{align*}
where the last equality follows from the i.i.d property. Thus, 
\begin{align*}
\frac{1}{n}\ln \mathbb{P}_{n} \Big(\sum^{n}_{j=1}Y_{j}\geq n\epsilon\Big) = - \theta \epsilon +  \ln ( e^{\theta}p + (1-p))  \leq  - \underset{\theta > 0 }{\max} \;\; \Big ( \theta \epsilon - \ln ( e^{\theta}p + (1-p))  \Big). 
\end{align*}
The function is concave in $\theta$ and the maximum is attained at, $\theta^* = \ln \Bigg(\dfrac{\epsilon(1-p)}{p(1-\epsilon)}\Bigg)$. Substituting the same in the above expression, 
\begin{align*}
\frac{1}{n}\ln \mathbb{P}_{n} \Big(\sum^{n}_{j=1}Y_{j}\geq n\epsilon\Big) \leq  \Big (\epsilon \ln \dfrac{p}{\epsilon} +   (1-\epsilon) \ln \dfrac{1-p}{1-\epsilon}\Big)
\end{align*}
For the proof of the ``greater than equal to " inequality we refer to \cite{varadhan2016large}. In the second observation model, the $\{Y_{n}\}$ sequence is not i.i.d. However, we would like to extend the same approach. The main challenge lies in finding the limit of the natural logarithm of the moment generating function. 
\begin{align*}
&\underset{n \to \infty}{\lim} \frac{1}{n}\ln \mathbb{E}_{\mathbb{P}_{n}}\Big[e^{\theta \big(\sum^{n}_{j=1}Y_{j}\big) }\Big] \\
&= \underset{n \to \infty}{\lim} \frac{1}{n}\ln \Bigg[ \sum^{n}_{s=0} \Bigg[ \sum_{\substack{\{y_{j}\}^{n}_{j=1} \in \{0,1\}^n: \\ \sum^{n}_{j=1}y_{j}=s}}\sum_{\{x_{j}\}^{n}_{j=1} \in \mathcal{X}^n} \Big[\prod^m_{k,l=1}\prod^m_{u,v=0,1} \Big(p^{v}_{l}\pi^{l}_{k,u} \Big)^{\eta(k,u,l,v)} \Big]\bar{\mathbb{P}}_{1}(x_{1}, y_{1}) \Bigg] e^{\theta s}\Bigg]\\
&\overset{(a)}{=} \underset{n \to \infty}{\lim} \frac{1}{n}\ln \Bigg[ \sum^{n}_{s=0} \Bigg[ \sum_{\substack{\{y_{j}\}^{n}_{j=1} \in \{0,1\}^n: \\ \sum^{n}_{j=1}y_{j}=s}}\Big(\sum^{m}_{l=1}p_{l}\pi^{*}_{l}\Big)^{s}  \Big(\sum^{m}_{l=1}(1-p_{l})\pi^{*}_{l}\Big)^{n-s}\Bigg] e^{\theta s}\Bigg]\\
&= \underset{n \to \infty}{\lim} \frac{1}{n}\ln \Bigg[ \sum^{n}_{s=0} \Big[ \prescript{n}{}{\mathbf{C}}_{s} e^{\theta s} \Big(\sum^{m}_{l=1}p_{l}\pi^{*}_{l}\Big)^{s}  \Big(\sum^{m}_{l=1}(1-p_{l})\pi^{*}_{l}\Big)^{n-s}\Big]\Bigg]\\
&=  \hspace{-4pt}\underset{n \to \infty}{\lim} \frac{1}{n}\ln \Bigg[ \Big(  e^{\theta s} \Big(\sum^{m}_{l=1}p_{l}\pi^{*}_{l}\Big) + \Big(\sum^{m}_{l=1}(1-p_{l})\pi^{*}_{l}\Big)\Big)^n\Bigg] \hspace{-4pt} = \hspace{-3pt} \ln \Bigg[  e^{\theta s} \Big(\sum^{m}_{l=1}p_{l}\pi^{*}_{l}\Big) + \Big(\sum^{m}_{l=1}(1-p_{l})\pi^{*}_{l}\Big)\Bigg],
\end{align*}
where equality $(a)$ follows from the asymptotic independence property.  Thus, 
\begin{align*}
\frac{1}{n}\ln \mathbb{P}_{n} \Big(\sum^{n}_{j=1}Y_{j}\geq n\epsilon\Big)  \leq - \underset{\theta > 0 }{\max} \;\; \Big ( \theta \epsilon -  \ln \Bigg[  e^{\theta s} \Big(\sum^{m}_{l=1}p_{l}\pi^{*}_{l}\Big) + \Big(\sum^{m}_{l=1}(1-p_{l})\pi^{*}_{l}\Big)\Bigg]  \Big). 
\end{align*}
Using the expression for the i.i.d case with $p= \sum^{m}_{l=1}p_{l}\pi^{*}_{l}$ we get, 
\begin{align*}
\frac{1}{n}\ln \mathbb{P}_{n} \Big(\sum^{n}_{j=1}Y_{j}\geq n\epsilon\Big)  \leq \epsilon \ln \dfrac{\sum^{m}_{l=1}p_{l}\pi^{*}_{l}}{\epsilon} +   (1-\epsilon) \ln \dfrac{1-\sum^{m}_{l=1}p_{l}\pi^{*}_{l}}{1-\epsilon} .
\end{align*}
\subsection{Discussion on Definition 1}\label{Discussion on Definition 1}
We note that $Z_{n} = \{ X_{n}, Y_{n}\}$ is a Markov Chain with transition matrix, $P$, given by, 
\begin{align*}
\begin{bmatrix}
\hspace{-10pt}&p_{1}\pi^{1}_{1,1} &p_{1}\pi^{1}_{2,1} \hspace{-8pt}& \ldots \hspace{-9pt}& p_{1}\pi^{1}_{m,1} &p_{1}\pi^{1}_{1,0} &p_{1}\pi^{1}_{2,0} \hspace{-8pt}&\ldots \hspace{-10pt}& p_{1}\pi^{1}_{m,0}\\
\hspace{-10pt}&p_{2}\pi^{2}_{1,1} &p_{2}\pi^{2}_{2,1} \hspace{-8pt}& \ldots \hspace{-9pt}& p_{2}\pi^{2}_{m,1} &p_{2}\pi^{2}_{1,0} &p_{2}\pi^{2}_{2,0} \hspace{-8pt}& \ldots \hspace{-10pt}& p_{2}\pi^{2}_{m,0}\\
\hspace{-10pt}&\ldots &\ldots \hspace{-8pt}&\ldots \hspace{-9pt}&\ldots &\ldots &\ldots \hspace{-8pt}&\ldots \hspace{-10pt}&\ldots\\
\hspace{-10pt}&p_{m}\pi^{m}_{1,1} &p_{m}\pi^{m}_{2,1} \hspace{-8pt}& \ldots \hspace{-9pt}& p_{m}\pi^{m}_{m,1} &p_{m}\pi^{m}_{1,0} &p_{m}\pi^{m}_{2,0} \hspace{-8pt}& \ldots \hspace{-10pt}& p_{m}\pi^{m}_{m,0}\\
\hspace{-10pt}&(1-p_{1})\pi^{1}_{1,1} &(1-p_{1})\pi^{1}_{2,1} \hspace{-8pt}& \ldots \hspace{-9pt}& (1-p_{1})\pi^{1}_{m,1} & (1-p_{1})\pi^{1}_{1,0} & (1-p_{1})\pi^{1}_{2,0} \hspace{-8pt}&\ldots \hspace{-10pt}&(1-p_{1})\pi^{1}_{m,0}\\
\hspace{-10pt}&(1-p_{2})\pi^{2}_{1,1} &(1-p_{2})\pi^{2}_{2,1} \hspace{-8pt}& \ldots \hspace{-9pt}&(1-p_{2})\pi^{2}_{m,1} &(1-p_{2})\pi^{2}_{1,0} &(1-p_{2})\pi^{2}_{2,0} \hspace{-8pt}&\ldots \hspace{-10pt}&(1-p_{2})\pi^{2}_{m,0}\\
\hspace{-10pt}&\ldots &\ldots \hspace{-8pt}&\ldots \hspace{-9pt}&\ldots &\ldots &\ldots \hspace{-8pt}&\ldots \hspace{-10pt}&\ldots\\
\hspace{-10pt}&(1-p_{m})\pi^{m}_{1,1} &(1-p_{m})\pi^{m}_{2,1} \hspace{-8pt}&\ldots \hspace{-9pt}& (1-p_{m})\pi^{m}_{m,1} &(1-p_{m})\pi^{m}_{1,0} &(1-p_{m})\pi^{m}_{2,0} \hspace{-8pt}&\ldots \hspace{-10pt}&(1-p_{m})\pi^{m}_{m,0}
\end{bmatrix}
\end{align*}
The stationary distribution of the Markov chain can be found by solving the linear equations, $\pi^{*^T} =  \pi^{*^T}P$, where $\pi^{*}$ belongs to the probability simplex in $\mathbb{R}^{2m}$. That is, 
\begin{align*}
&\pi^{*}_{k} = \sum^{m}_{l=1}\pi^{*}_lp_{l}\pi^{l}_{k,1} +  \sum^{m}_{l=1}\pi^{*}_{m+l}p_{l}\pi^{l}_{k,0}, 1 \leq k \leq m\\
&\pi^{*}_{m+k} = \sum^{m}_{l=1}\pi^{*}_l(1-p_{l})\pi^{l}_{k,1} +  \sum^{m}_{l=1}\pi^{*}_{m+l}(1-p_{l})\pi^{l}_{k,0}, 1 \leq k \leq m
\end{align*}
\begin{proposition}
Suppose $\bar{\pi}$ is the stationary distribution of a Markov chain with the transition matrix,
\begin{align*}
\bar{P} = \begin{bmatrix}
&p_{1}\pi^{1}_{1,1} + (1-p_{1})\pi^{1}_{1,0} &p_{1}\pi^{1}_{2,1} + (1-p_{1})\pi^{1}_{2,0} &\ldots &p_{1}\pi^{1}_{m,1} + (1-p_{1})\pi^{1}_{m,0}\\
&p_{2}\pi^{2}_{1,1} + (1-p_{2})\pi^{2}_{1,0} &p_{2}\pi^{2}_{2,1} + (1-p_{2})\pi^{2}_{2,0} &\ldots &p_{2}\pi^{2}_{m,1} + (1-p_{2})\pi^{2}_{m,0}\\
&\ldots &\ldots  &\ldots  &\ldots \\
&p_{m}\pi^{m}_{1,1} + (1-p_{m})\pi^{m}_{1,0} &p_{m}\pi^{m}_{2,1} + (1-p_{m})\pi^{m}_{2,0} &\ldots &p_{m}\pi^{m}_{m,1} + (1-p_{m})\pi^{m}_{m,0}\\
\end{bmatrix}.
\end{align*} 
Then, $\pi^{*}_{k} = p_{k}\bar{\pi}_{k}, \pi^{*}_{m+k} = (1-p_{k})\bar{\pi}_{k}, 1\leq k \leq m$. 
\end{proposition}
\begin{proof}
From the stationary property of $\bar{\pi}$, it follows that 
\begin{align*}
&\bar{\pi}_{k} = \sum^{m}_{l=1}\Big(p_{l}\pi^{l}_{k,1} + (1-p_{l})\pi^{l}_{k,0}\Big)\bar{\pi}_{l} \\
&\implies p_{k}\bar{\pi}_{k} = \sum^{m}_{l=1}\big(p_{l}\bar{\pi}_{l}\big)p_{k}\pi^{l}_{k,1} + \sum^{m}_{l=1} \big( (1-p_{l})\bar{\pi}_{l} \big) p_{k}\pi^{l}_{k,0}, \\
&(1 -p_{k})\bar{\pi}_{k} = \sum^{m}_{l=1}\big(p_{l}\bar{\pi}_{l}\big)(1-p_k)\pi^{l}_{k,1}+ \sum^{m}_{l=1} \big((1-p_{l})\bar{\pi}_{l}\big)(1-p_k) \pi^{l}_{k,0}. \\
&\implies p_{k}\bar{\pi}_{k} = \sum^{m}_{l=1}\big(p_{l}\bar{\pi}_{l}\big)p_{l}\pi^{l}_{k,1} +  \sum^{m}_{l=1}\big(p_{l}\bar{\pi}_{l}\big)\big(p_k-p_{l}\big)\pi^{l}_{k,1}  + \sum^{m}_{l=1} \big( (1-p_{l})\bar{\pi}_{l} \big) p_{l}\pi^{l}_{k,0} \\  &+\sum^{m}_{l=1} \big( (1-p_{l})\bar{\pi}_{l} \big) \big(p_{k}-p_{l}\big)\pi^{l}_{k,0},\\
&(1 -p_{k})\bar{\pi}_{k} = \sum^{m}_{l=1}\big(p_{l}\bar{\pi}_{l}\big)(1-p_l)\pi^{l}_{k,1}+ \sum^{m}_{l=1}\big(p_{l}\bar{\pi}_{l}\big)(p_l -p_k)\pi^{l}_{k,1} + \sum^{m}_{l=1} \big((1-p_{l})\bar{\pi}_{l}\big)(1-p_l) \pi^{l}_{k,0} \\
& + \sum^{m}_{l=1} \big((1-p_{l})\bar{\pi}_{l}\big)(p_l-p_k) \pi^{l}_{k,0}. \\
&\implies p_{k}\bar{\pi}_{k}  + (1 -p_{k})\bar{\pi}_{k} =  \underbrace{\sum^{m}_{l=1}\big(p_{l}\bar{\pi}_{l}\big)p_{l}\pi^{l}_{k,1}  + \sum^{m}_{l=1} \big( (1-p_{l})\bar{\pi}_{l} \big) p_{l}\pi^{l}_{k,0}} + \\
&\hspace{6cm} \underbrace{\sum^{m}_{l=1}\big(p_{l}\bar{\pi}_{l}\big)(1-p_l)\pi^{l}_{k,1} +  \sum^{m}_{l=1} \big((1-p_{l})\bar{\pi}_{l}\big)(1-p_l) \pi^{l}_{k,0} }
\end{align*}
If we define, $\hat{\pi}_{k} = p_{k}\bar{\pi}_{k},\; 1\leq k \leq m$ and $\hat{\pi}_{m+k} = (1 -p_{k})\bar{\pi}_{k}, \;  1\leq k \leq m$, and, we assign, 
\begin{align*}
&\hat{\pi}_{k} = \sum^{m}_{l=1}\hat{\pi}_{l}p_{l}\pi^{l}_{k,1}  + \sum^{m}_{l=1}\hat{\pi}_{m+l} p_{l}\pi^{l}_{k,0},\; 1\leq k \leq m \\
&\hat{\pi}_{m+k} = \sum^{m}_{l=1}\hat{\pi}_{l}(1-p_l)\pi^{l}_{k,1} +  \sum^{m}_{l=1} \hat{\pi}_{m+l} (1-p_l) \pi^{l}_{k,0} , \; 1\leq k \leq m ,
\end{align*}
we retrieve the equations satisfied by the stationary distribution $\pi^{*}$. 
\end{proof}
We note that in the above proof, there is a lack of uniqueness is expressions satisfied by $\{\hat{\pi}_{k}, \hat{\pi}_{m+k}\}$. The expressions in R.H.S could be modified by adding terms in the expression for $\hat{\pi}_{k}$ and subtracting the same terms in the expression of $\hat{\pi}_{m+k}$.   
\begin{proposition}
If there exists $\bar{\pi}$ such that $\pi^*_{k} = p_{k}\bar{\pi}_{k}$ and $\pi^*_{k} =(1- p_{k})\bar{\pi}_{k}$, then $\bar{\pi}$ is the stationary distribution of a Markov chain with transition matrix $\bar{P}$. 
\end{proposition}
\begin{proof}
Substituting $\pi^*_{k} = p_{k}\bar{\pi}_{k}$ and $\pi^*_{k} =(1- p_{k})\bar{\pi}_{k}$ in the equations corresponding to stationarity of $P$
\begin{align*}
&p_{k}\bar{\pi}_{k} = \sum^{m}_{l=1}p_{l}\bar{\pi}_{l}p_{l}\pi^{l}_{k,1} +  \sum^{m}_{l=1}(1- p_{l})\bar{\pi}_{l}p_{l}\pi^{l}_{k,0}, 1 \leq k \leq m,\\
&(1- p_{k})\bar{\pi}_{k} = \sum^{m}_{l=1}p_{l}\bar{\pi}_{l}(1-p_{l})\pi^{l}_{k,1} +  \sum^{m}_{l=1}(1- p_{l})\bar{\pi}_{l}(1-p_{l})\pi^{l}_{k,0}, 1 \leq k \leq m.
\end{align*}
Adding the above two expressions, we get, 
\begin{align*}
\bar{\pi}_{k}  &= \sum^{m}_{l=1} \big(p^{2}_l + p_{l} - p^{2}_l\big)\bar{\pi}_{l}\pi^{l}_{k,1} +  \sum^{m}_{l=1} \big(p_{l} -p^{2}_{l} + 1- 2p_l + p^{2}_{l}  \big)\bar{\pi}_{l}\pi^{l}_{k,0}\\
&= \sum^{m}_{l=1} p_{l}\bar{\pi}_{l}\pi^{l}_{k,1} +  \sum^{m}_{l=1} (1-p_l)\bar{\pi}_{l}\pi^{l}_{k,0}=  \sum^{m}_{l=1}\Big(  p_{l}\pi^{l}_{k,1} + (1-p_l)\pi^{l}_{k,0}\Big)\bar{\pi}_{l}
\end{align*}
\end{proof}
The asymptotic independence property of the joint distribution can be interpreted as follows. For arbitrarily large $n$, the joint probability of $\{Y_{j}=y_j\}^{n}_{j=1}$ can be computed by using the stationary distribution of the Markov chain and numbers of $1$'s and $0$'s in $\{y_j\}^{n}_{j=1}$. 
\subsection{Pseudo-code for Classification Algorithm and Proof of Proposition 5} \label{Safety Algorithm}
The pseudo code for the classification algorithm is presented in Algorithm \ref{Algorithm 2}. In the first part of the algorithm, from lines $1-12$, the threshold is computed. The threshold based classification algorithm is described between lines $13-19$. Below we present the proof of Proposition \ref{Proposition 5}. The proof relies on the idea that after a finite number of visits to all the regions, the agent is able to estimate the parameter corresponding  to each region to a precision where the index corresponding to the threshold is identified and does not change for future iterations. Thus, the label corresponding to each agent is identified and does not change with further iterations. The further iterations are only to improve the  accuracy of the parameter estimates for the safe regions. 
\begin{algorithm}
\caption{Safety Classification Algorithm}
\begin{algorithmic}[1]\label{Algorithm 2}
\STATE Given $\{c_{k}\}^m_{k=1}$ in nondecreasing order, $\max \gets 0$, $k \gets 1$, 
\WHILE {$k \leq m-1 $}
\IF {$c_{k+1} - c_{k} > \max $}
\STATE $\max \gets c_{k+1} - c_{k} $, $k^* \gets k$
\ENDIF
\STATE $k \gets k+1  $
\ENDWHILE
\IF {$k^* \in  [\lceil \frac{m}{4}\rceil, \lceil \frac{3m}{4}\rceil] $}
\STATE $c_{T} = c_{k^*}$
\ELSE
\STATE $c_{T} = median(\{c_{l}\}^{m}_{l=1})$ 
\ENDIF
\WHILE {$k \leq m$}
\IF    {$ c_{k} \leq c_{T}$ and $p_{k} \geq \frac{1}{2}$}
\STATE $LABEL[k] \gets SAFE$
\ELSE 
\STATE $LABEL[k] \gets UNSAFE$
\ENDIF
\ENDWHILE
\end{algorithmic}
\end{algorithm}
\begin{proof}
By the strong law of large numbers, it follows that there exists $N_{\omega, \delta}$ such that $|p_{k,n} - p^*_k|< \delta, \forall n \geq N_{\omega, \delta}, k \in \{1, \ldots, m\}$, i.e., after a finite number of visits to all the cells the estimate of the probability of observing $1$ at a given cell becomes $\delta$ close to the true probability for every cell. Suppose $\delta$ is smaller than $|p^{*}_{k}(1-p^*_k) - p^{*}_{l}(1-p^*l)|$ for any $k,l$. Then $\{(p_{k, N_{\omega}})(1- p_{k, N_{\omega}})\}^{m}_{k=1}$ can be sorted in ascending order and the threshold $c_{T}$ can be obtained. From the definition of the thresholds, it follows that that index corresponding to the threshold does not change after $N_{\omega,\delta}$. Hence, the set $\mathcal{X}_s$ obtained at $N_{\omega, \delta}$ remains invariant for all future iterations and states of the MDP remain in this set for all $n \geq N_{\omega, \delta}$.  
\end{proof}
\end{document}